\documentclass{article}
\usepackage[english]{babel}

\usepackage{amssymb, amsmath, wasysym, mathrsfs}
\usepackage{graphicx, wrapfig}
\usepackage{tabularx}
\usepackage{microtype}
\usepackage{enumitem}
\usepackage{latexsym,xcolor,amsthm}
\usepackage[colorlinks]{hyperref}

\title{Critical Placements of a Square or Circle amidst Trajectories for Junction Detection}

\author{Ingo van Duijn\thanks{MADALGO, Aarhus University,
  \texttt{ivd@cs.au.dk}}
  \and
Irina Kostitsyna\thanks{Universit\'{e} libre de Bruxelles,
  \texttt{irina.kostitsyna@ulb.ac.be}}
  \and
Marc van Kreveld\thanks{Utrecht University
    \texttt{m.j.vankreveld@uu.nl}}
  \and
Maarten L\"offler\thanks{Utrecht University
    \texttt{m.loffler@uu.nl}}
}

\graphicspath{{figures/}}

\newcommand{\Oh}{\mathcal{O}}
\newcommand{\Li}{\mathcal{L}}
\newcommand{\arr}{{\cal A}}
\newcommand{\eps}{\varepsilon}
\newcommand{\T}{{\cal T}}
\newcommand{\Sq}{\Box}
\newcommand{\Ci}{\bigcirc}
\newcommand{\tr}{{\tau}}
\newcommand{\rlem}[1]{Lemma~\ref{lem:#1}}
\newcommand{\rfig}[1]{Figure~\ref{fig:#1}}
\newcommand{\rsec}[1]{Section~\ref{sec:#1}}
\newcommand{\Se}[2]{s(#2)}
\newcommand{\Seg}[2]{S(#1, #2)}
\newcommand{\Segs}[1]{S(#1)}

\newtheorem{theorem}{Theorem}

\newtheorem{lemma}[theorem]{Lemma}
\newtheorem{cor}[theorem]{Corollary}

\newtheorem{defi}{Definition}

\date{}
\begin{document}

\maketitle

\begin{abstract}
Motivated by automated junction recognition in tracking data, we study a problem of placing a
square or disc of fixed size in an arrangement of lines or line segments in the plane.
We let distances among the intersection points of the lines and line segments with the square
or circle define a clustering, and study the complexity of \emph{critical} placements for this clustering.
Here critical means that arbitrarily small movements of the placement change the clustering.

A parameter $\eps$ defines the granularity of the clustering. Without any assumptions on $\eps$,
the critical placements have a trivial $\Oh(n^4)$ upper bound. When the square or circle has
unit size and $0 < \eps < 1$ is given, we show a refined $\Oh(n^2/\eps^2)$ bound,
which is tight in the worst case.

We use our combinatorial bounds to design efficient algorithms to compute junctions.
As a proof of concept for our algorithms we have a prototype implementation that showcases
their application in a basic visualization of a set of $n$ trajectories and their $k$ most important junctions.
\end{abstract}

\section{Introduction}

Many analysis problems in geography have an inherent scale component:
the ``granularity'' or ``coarseness'' at which the data is studied. 
The most direct way to model spatial scale in geographic problems is by using a
fixed-size neighborhood of locations, such as a fixed-size square or circle.
For example, \emph{population density} can be studied at the scale of a city or
at the scale of a country, where one may consider the population in units with
an area of $10^4\; m^2$ or $25\; km^2$, respectively.
There are many other cases where local geographic phenomena are studied at
different spatial scales.
The Geographical Analysis Machine is an example of a system that supports
such analyses on point data sets~\cite{openshaw}.

In computational geometry, the problem of computing the placement of a square or circle
to optimize some measure has received considerable attention.
For a set of $n$ points in the plane, one can compute the (fixed-size, fixed-orientation)
square that maximizes the number of points inside in $\Oh(n \log n)$ time
(expand every point to a square, and the problem becomes finding
a point in the maximum number of squares which is solved by a sweep).
Mount et al.~\cite{mount} study the overlap function of two simple polygons under translation and
show, among other things, that one can compute the placement of a square that maximizes the area of
overlap with a simple polygon with $n$ vertices in $\Oh(n^2)$ time.
For a weighted subdivision, one can compute a placement that maximizes the weighted area inside
in $\Oh(n^2)$ time as well~\cite{Buchin}; this problem is motivated by clustering in aggregated data.
In the context of diagram placement on maps, various other measures to minimize or maximize when
placing squares were considered, like the total length of border overlap~\cite{ksw}.

Our interest lies in a problem concerning trajectory data. A trajectory is represented
by a sequence of points with associated time stamps, and models the movement of an entity through
space; we will assume in this paper that the movement space is two-dimensional.
The identification of ``interesting regions'' in the plane defined by a collection of trajectories
has  been studied in several papers recently.
These regions can be characterized as meeting places~\cite{gk}, popular or interesting
places~\cite{BenkertDGW10,gks,PalmaBKA08,TiwariK13}, and
stop regions~\cite{MorenoTRB10}.
In several cases, interesting regions are also defined as squares of fixed size, placed suitably.
The more algorithmic papers show such regions of interest can typically be computed in $\Oh(n^2)$ or $\Oh(n^3)$ time;
what is possible of course depends on the precise definition of the problem.
There are many other types of problems that can be formulated with trajectory data.
For overviews, see~\cite{glw-mpstd-08,jyj-tpm-11}.

Besides exact algorithms for optimal square or circle placement, approximation algorithms have been
developed for several problems on trajectory and other data, see e.g.~\cite{gk,hm,hk}.
\medskip


\noindent
{\bf Motivation and problem description.}
We consider a problem on trajectories related to common movements and
changes of movement directions at certain places. Imagine a large open space like a
town square, a large entrance hall, or a grass field. People tend to traverse such areas in ways
that are not random, and the places where a decision is made and possibly a change of direction
is initiated may be specific. Also for data like ant tracks, the identification of places where
tracks go different ways is of interest.
Without going into details, these observations motivate us to study placements of a square or
circle of fixed size such that bundles of incoming and outgoing entities arise.

\begin{figure}[tb]
\begin{center}
\includegraphics[width=.9\textwidth]{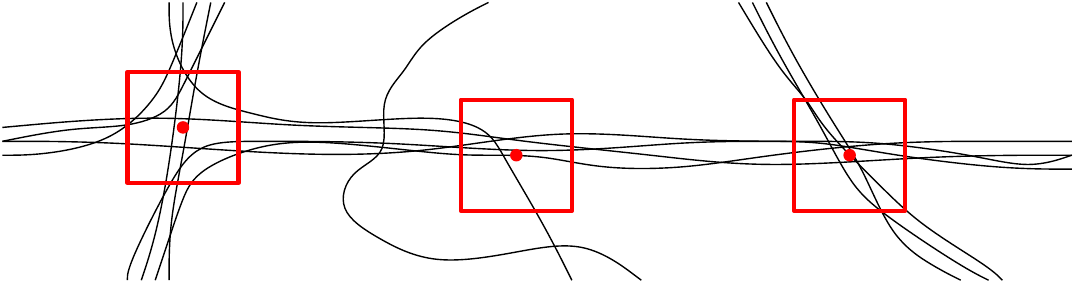}
\end{center}
\caption{A set of trajectories and square placements at junctions of varying significance.}
\label{fig:example}
\end{figure}

Consider the tracks in Figure~\ref{fig:example}. We observe that to define places where the tracks of entities cross and where decisions are made, we can use the placement of a square of a certain size
and where the tracks enter and leave the square. We are interested in placements that give rise to bundles: large subsets of tracks that enter and leave the square at a similar place on its boundary, and with a gap to the next location along the boundary where this happens.
The left square in the figure has five bundles where one bundle consists of only one trajectory, and the middle and right square have only two bundles (albeit with different ``topology'').
We see that the left and right squares indicate regions that should be considered more significant than the middle one for being a junction.
The main difference between the left and right junctions is that on the left, decisions are made to go straight or change direction, whereas on the right, all moving entities went straight and no different decisions were taken.
\smallskip

We study an abstracted version of this placement problem and ignore many of the
practical issues that our simplified definition of a junction would have.
Later in the paper we briefly address such issues by giving a slightly more
involved definition of a junction.
The majority of the paper concentrates on an abstract
combinatorial and computational problem that
lies at the core of junction detection.

Let $\T$ be a collection of trajectories, let $\Sq_p$ be the boundary of a unit square $S$
centered at $p$ placed axis-aligned amidst the trajectories, and let $\eps$ be a positive real constant.
Let $Q=\T\cap\Sq_p$ be the set of all intersections of trajectories with the
boundary of the square (here we can ignore the time component of trajectories; they are considered polygonal lines).
Two points in $Q$ are \emph{$\eps$-close} if their distance along $\Sq_p$ is at most $\eps$.
$\T$ and $\Sq_p$ give rise to a \emph{clustering} of $Q$ on $\Sq_p$ by the transitive closure
relation of $\eps$-closeness.
Different clusters are separated by a distance larger than $\eps$ along $\Sq_p$.
A single cluster of $Q$ corresponds to parts of trajectories that enter or leave $S$ in
each other's proximity (in Figure~\ref{fig:example} from left to right the squares define four, three, and four clusters respectively).

Now consider the two-dimensional space of all placements of a square of fixed size by choosing
its center as its reference point.
We say that a placement $q$ is \emph{critical}, if any arbitrarily small neighborhood of $q$
contains points inducing different clusters.
This can be due to a change in size of a cluster on $\Sq_q$ or to a change in the clustering.
The latter corresponds to placements where the distance between two points of
$\T\cap \Sq_p$ on $\Sq_p$ is exactly $\eps$, with no other points of $\T\cap \Sq_p$
in between. Since noncritical points are part of a region that defines the same clustering,
one can think of the set of critical points to be the boundary between these regions.
\medskip

\noindent
{\bf Results and organization.}
In Section~\ref{sec:complex} we analyze the complexity of the space of critical placements
for fixed-size squares in an arrangement of $n$ lines.
We show that the placement space has $\Oh(n^2/\eps^2)$ total complexity in the worst case,
and can be constructed in $\Oh(n^2\log n + k)$ time where $k$ is the true complexity (the latter appearing in the appendix due to space constraints).
In Section~\ref {sec:lower}, we show that these results are tight by presenting an explicit construction that exhibits the worst-case behavior.
In Section~\ref{sec:junction} we discuss our application to junction detection further and show
output from a prototype implementation.
We conclude in Section~\ref{sec:conclusions}.
In the appendix we further show how to extend our approach to more realistic settings,
such as placing a square on an arrangement of line segments or placing a circle rather than a square.

\section{Square on lines}
\label{sec:complex}

We begin by studying the simplest version of the problem: placing a unit square over an arrangement of lines.
The lines ``cut'' the boundary of the square into several pieces.
We are interested in all placements of the square such that one of these pieces has length exactly $\eps$,
in which case we call the piece an \emph {$\eps$-segment} and the placement \emph {critical}.
When a piece contains one of the corners of the square, its length is the sum of its two incident line segments.
Note that this definition of a critical placement is a simplification of the one defined earlier in terms of clusters;
it only considers merging and splitting clusters.
In the definition from the introduction, a placement is also critical if a corner of the square coincides with a line.
However, these placements are simply four translates of the input lines, so the rest of this section focuses on the harder critical placements as defined here.

\subsection{Placement space}


Let $\Li$ be a set of lines and denote by $\arr$ the arrangement of $\Li$.
For a placement of the square on $\arr$, consider all cells that contain part of its boundary.
To characterize all critical placements, we look at how the square can be ``moved around'' such that a given cell of $\arr$ contains an $\eps$-segment throughout the motion.
For instance, in \rfig{definitions} on the left, we can move the square slightly left and right such that there is always an $\eps$-segment in the same cell.
We use the intuition of moving the square to argue that the critical placements can be characterized as a set of line segments, and
then prove an upper bound on how many times these line segments can intersect.

\begin{defi}
\label{def:epsplace}
Let $\Li$ be a set of lines and $\Sq_p$ be the boundary of an axis-aligned unit square whose center is denoted by $p$, and let $\eps>0$ be a constant. A placement of $\Sq_p$ (or $p$) is an \emph{$\eps$-placement}
if at least one connected component in $\Sq_p \setminus \Li$ has length exactly $\eps$.
An \emph{$\eps$-segment} is a connected component of $\Sq_p \setminus \Li$ with
length exactly $\eps$.
\end{defi}

Consider the example from the figure again.
When moving $p$ to the right we maintain the indicated $\eps$-segment; this reduces $p$'s movement to one degree of freedom.
It can happen that when moving $p$,
another $\eps$-segment is created in another cell.
This means that this $\eps$-placement gives rise to two $\eps$-segments.
We assume $\Li$ to be in general position such that no two $\eps$-segments give $p$ the same allowed movement (a condition that is met after perturbing the input).
Therefore no more than two $\eps$-segments will ever occur simultaneously.

Since $\eps$-segments are part of $\Sq_p$, it is convenient to fix a point $p'$ on $\Sq_p$ and look at the movement of $p'$ instead of $p$.
Thus, by moving the square such that an $\eps$-segment is maintained inside a cell $c$ in $\arr$, the point $p'$ traces a curve.
If we consider only those parts of this curve corresponding to placements where $p'$ lies on the $\eps$-segment,
we observe that this subcurve is contained in $c$.
For instance, in \rfig{definitions} on the bottom right, the fixed point $p'$ can be moved vertically ($p$ and the square move accordingly) exactly between the intersection points with $\Li$.
To facilitate our analysis, we will choose a set of fixed points on $\Sq_p$ such that any $\eps$-segment will contain exactly one of these points.
For ease of presentation we assume that $\frac 1 \eps$ is integer, so that we can place exactly $\frac{1}{\eps}$ fixed points with distance $\eps$ apart along $\Sq_p$.
If it is not integer, we can pick fixed points such that an $\eps$-segment contains one or two such points, in which case our analysis overcounts by a constant factor.

Since these points are fixed with respect to $p$, we define them in terms of \emph{translation vectors}.
We write $\tr(X)$ for the translation by $\tr$ of any object $X$. The inverse of $\tr$ is denoted $\tr^{-1}$.
Thus, in \rfig{definitions}, $p' = \tr(p)$.

\begin{defi}
Let $\frac{1}{\eps}$ be integer. $T_\eps$ is the set of $\frac{4}{\eps}$ vectors such that $\Sq_p \setminus \{\tr(p) \mid \tr \in T_\eps\}$
is a set of open line segments each of length $\eps$, and $T_\eps$ includes all vectors $\sigma$ such that $\sigma(p)$ is a corner of $\Sq_p$.
\end{defi}

\begin{figure}[tb]
\begin{center}
\includegraphics[width=.95\textwidth]{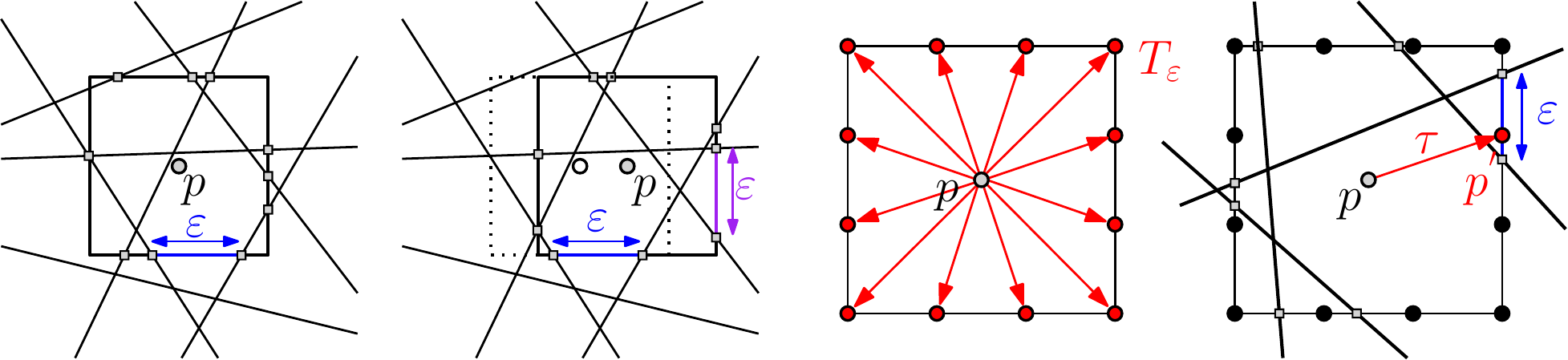}
\end{center}
\caption{Left two: placement of square with one $\eps$-segment and two $\eps$-segments. Right two: the set of vectors $T_\eps$ and an $\eps$-placement $p$ with the vector in $T_p$ indicated.}
\label{fig:definitions}
\end{figure}

Let $c$ be a cell of $\arr$ and $\tau \in T_\eps$ a vector.
We denote by $\Seg{c}{\tr}$ the set of all critical placements $p$ such that $\tau(p)$ lies on an $\eps$-segment in $c$.
That is, $\tau(\Seg{c}{\tau})$ is the set of curves traced out by $\tau(p)$ under our previous interpretation of moving $p'$.
We note that some cells are ``too small'' to contain an $\eps$-segment, in which case $\Seg{c}{\tr}$ is empty; other cells may contain up to two disconnected curves (see Figure~\ref{fig:sctau}).
As shorthand notation, let $\Segs{\tr}$ denote the union
of all $\Seg{c}{\tr}$ over all cells in $\arr$.

If $\tau(p)$ is not a corner of $\Sq_p$, then $\tau(\Seg{c}{\tr})$ coincides with (one\footnote{Degenerate case where the width or height of $c$ is exactly $\eps$.} or) two parallel axis-aligned $\eps$-segments contained in $c$.
Therefore, the number of such segments is bounded by twice the number of cells in the arrangement of lines.
If $\tau(p)$ \emph {is} a corner of $\Sq_p$, the shape of $\Seg{c}{\tr}$ is more complex. This means that a similar bound on $\Segs{\tr}$ as before is not as easy to achieve.
\rfig{sctau} shows the construction of $\Seg{c}{\tr}$ when $\tr(p)$ is the upper right corner of $\Sq_p$.
The curve inside $c$ shows exactly how $\tr(p)$ can be moved such that $c$ contains an $\eps$-segment containing $\tr(p)$.
Translated by the inverse of $\tr$, we get the actual $\eps$-placements of $p$ corresponding to the particular cell $c$ and vector~$\tr$.
Note that the figure also shows the only case when $\Seg{c}{\tr}$ can be disconnected, i.e. when $c$ contains an acute angle in the same ``direction'' as $\tau$.
We will show that for the four corner vectors $\tr$, the curves in $\Segs{\tr}$ have properties that
allow us to bound the complexity of the space of all $\eps$-placements.

\begin{figure}[tb]
\centering
    \includegraphics[scale=.8]{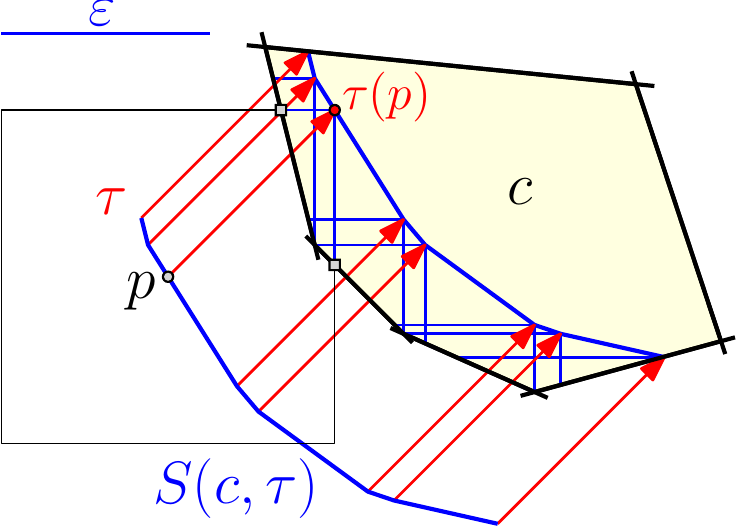}
    \hfil
    \includegraphics[scale=.8]{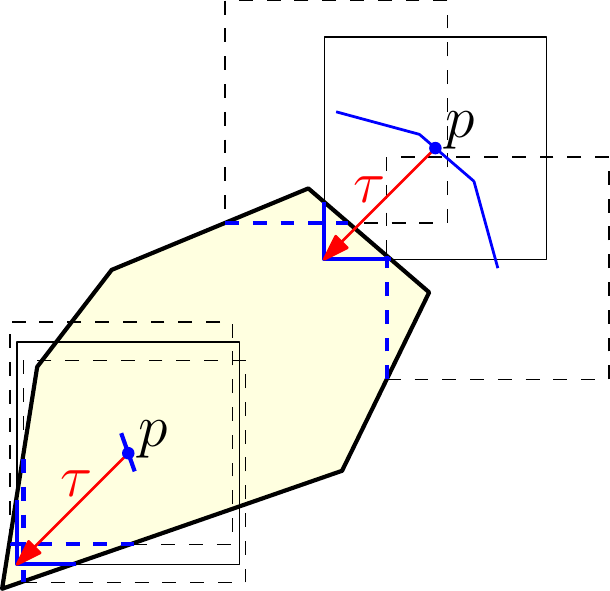}
    \caption{Left: construction of the curve $\Seg{c}{\tr}$. Right: $\Seg{c}{\tr}$ can consist of multiple pieces.}
    \label{fig:sctau}
\end{figure}


\begin{lemma}\label{lem:convex}
For $\tr \in T_\eps$ and a cell $c$ in $\arr$, $\Seg{c}{\tr}$ consists of at most two connected subsets of a piecewise linear and convex curve.
\end{lemma}
\begin{proof}
If $\tr(p)$ is not a corner point of $\Sq_p$ then $\Seg{c}{\tr}$ is either empty or consists of two single line segments.

Without loss of generality let $\tr$ translate
$p$ to the upper right corner of $\Sq_p$. Because the result must hold for arbitrary $\eps$, we define the following function $f$ related
to the notion of an $\eps$-segment, but without being dependent on $\eps$.

For a point $a$, let $x_a$ be the distance to the closest line in $\Li$ left of $a$,
and let $y_a$ be the distance to the closest line below $a$; define $f(a) = x_a+y_a$.
Let $f_c$ be $f$ restricted to the points in a cell $c$ of $\arr$, and consider
two points $a$ and $b$ inside $c$. For a point $v$ halfway between $a$ and $b$,
we have $x_v \geq (x_a+x_b)/2$ and $y_v\geq (y_a+y_b)/2$ by the convexity of $c$.
It follows that $f_c(v)\geq (f_c(a)+f_c(b))/2$. 
Thus, $f_c$ is a concave function inside $c$, and
taking the level set of a concave function yields a convex curve (possibly extending outside $c$).
Observe that the set of all points $a \in c$ such that $f(a) = \eps$ is the same as $\Seg{c}{\tr}$
and is a level set of $f_c$, so it is convex.
Similarly, $f_c$ is piecewise linear and hence, so is $\Seg{c}{\tr}$.
Since there are at most $4$ points on the boundary of $c$ where $f_c$ has value $\eps$,
$\Seg{c}{\tr}$ consists of at most two connected components.
\end{proof}

\subsection {Complexity}

\begin{lemma}\label{lem:sstar}
For all $\tr \in T_\eps$, $\Segs{\tr}$ consists of $\Oh(n^2)$ line segments.
\end{lemma}
\begin{proof}
For any $\tau$ not corresponding to a corner, $S(\tau)$ contains at most two horizontal or two vertical
line segments for each cell of $\arr$. Next, assume that $\tau$ corresponds to a corner of $\Sq_p$, say,
the upper right corner.
Let $c$ be a cell with $k$ edges in $\arr$, and let $p$ be an $\eps$-placement such that $\tr(p) \in c$.
For a fixed $y$-coordinate of $p$, there are at most two $\eps$-placements with $\tr(p) \in c$ by convexity of the cell $c$.
Furthermore, for each vertex of $\Seg{c}{\tr}$, the point $\tr(p)$ aligns horizontally or vertically with
a vertex of $c$. It follows that the complexity of $\Seg{c}{\tr}$ is $\Oh(k)$.
%
Summing over the complexities of all cells
in the arrangement gives the bound.
\end{proof}

We are interested in the complexity of the placement space of $\Sq_p$, and it is therefore important to know how many times
two sets of curves $\Segs{\tr}$ and $\Segs{\sigma}$ intersect, for $\tr,\sigma\in T_\eps$.
An intersection of these two sets corresponds to a placement of $\Sq_p$
such that $\tr(p)$ and $\sigma(p)$ both lie on an $\eps$-segment of $\Sq_p$. The following lemma shows that
the number of intersections is bounded by the complexity of $\arr$.

\begin{lemma}\label{lem:nsquared}
For distinct $\tr, \sigma \in T_\eps$, the intersection of $\Segs{\tr}$ and $\Segs{\sigma}$ contains $\Oh(n^2)$ points.
\end{lemma}
\begin{figure*}[htb]
\begin{center}
\includegraphics[width=.99\textwidth]{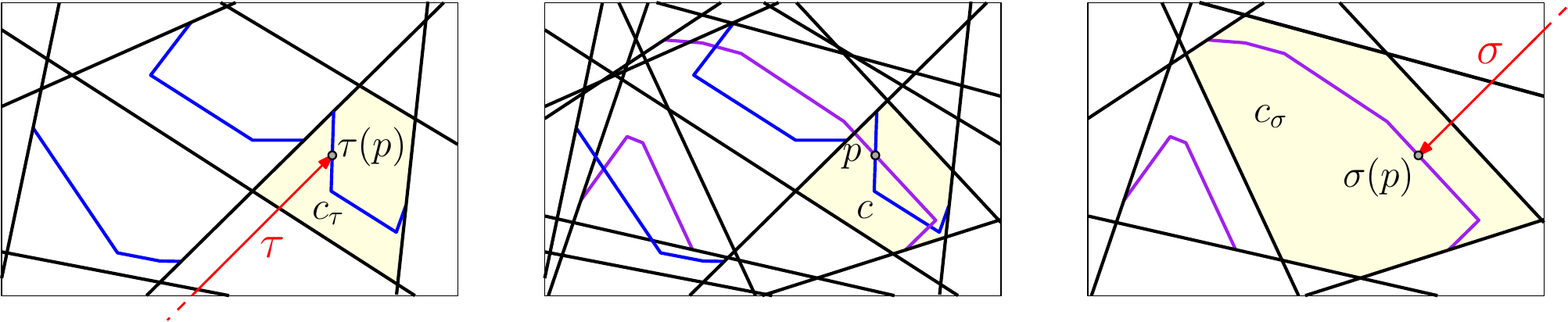}
\end{center}
\caption{Example arrangement with $\eps$-placements for two vectors $\tr$ and $\sigma$.
Left: $\Li$ and $\tr(\Segs{\tr})$. Right: $\Li$ and $\sigma(\Segs{\sigma})$. Center: Both
 combined under appropriate translations. }
\label{fig:refinedarr}
\end{figure*}
%
%
\begin{proof} Let $p$ be an $\eps$-placement such that $p \in \Segs{\tr} \cap \Segs{\sigma}$, for distinct $\tr, \sigma \in T_\eps$. Let
$\arr' = \arr \left( \tr^{-1}(\Li) \cup \sigma^{-1}(\Li) \right) $ be the arrangement of two copies of $\Li$ translated
by the inverses of $\tau$ and $\sigma$. Let $c$ denote the cell in $\arr'$ that contains $p$, and
let $c_\tr$ and $c_\sigma$ denote the cells in $\arr$ that respectively
contain $\tr(p)$ and $\sigma(p)$  (see \rfig{refinedarr}).
Observe that $c = \tr^{-1}(c_\tr) \cap \sigma^{-1}(c_\sigma)$.

Let $\Se{p}{\tr}$ denote the line segment in $\Seg{c_\tr}{\tr}$ on which $p$ lies.
Distinguish the following two cases for $p$: either at least one end point of $\Se{p}{\tr}$ or $\Se{p}{\sigma}$ lies in $c$, or none. In the first case,
without loss of generality assume that one end point of $\Se{p}{\tr}$ lies in c.
By \rlem{convex}, $\Seg{c_\sigma}{\sigma}$ consists of at most two convex curves, hence $\Se{p}{\tr}$ intersects $\Seg{c_\sigma}{\sigma}$ at most four times.
 Because $\Seg{c_\sigma}{\sigma}$ is the only part of $\Segs{\sigma}$
intersecting $c$, it also holds that the part of $\Se{p}{\tr}$ inside $c$ intersects $\Segs{\sigma}$ only four times.
Thus, all intersections of this kind are bounded by the number of vertices in $\Segs{\tau}$, which by \rlem{sstar} is $\Oh(n^2)$

In the other case, no end point of $\Se{p}{\tr}$ or $\Se{p}{\sigma}$ lies in $c$.
Consider an edge $e$ of $c$ in $\arr'$ that intersects $\Se{p}{\tr}$.
Since $e$ intersects $\Se{p}{\tr}$ it must be an edge of $c_\sigma$.
Therefore, it intersects $\Seg{c_\tr}{\tr}$ and thus $\Segs{\tr}$ at most twice.

Therefore, the number of intersections of this kind is bounded by the number of edges in $\arr'$, which is $\Oh(n^2)$.
%
\end{proof}
Since there are $\Oh(\frac 1 {\eps^2})$ pairs of vectors in $T_\eps$, we immediately obtain:
\begin{theorem}
Given a set of $n$ lines $\Li$, the arrangement of $\eps$-placements of a unit square
has complexity $\Oh(\frac{n^2}{\eps^2})$.
\end{theorem}

\section {Extensions}
\label {sec:ext}

Here we describe how to extend the general approach when placing a square on lines to different, more realistic settings.
We describe what happens when we replace the line arrangement with an arrangement of line \emph {segments}, denoted $\arr(\T)$.
Additionally, we describe how to adapt the approach when placing a circle rather than a square.

\subsection{Square on line segments}

\begin{wrapfigure}[10]{r}{0.3\textwidth}
  \vspace{-1.8\baselineskip}
    \includegraphics[width=0.3\textwidth]{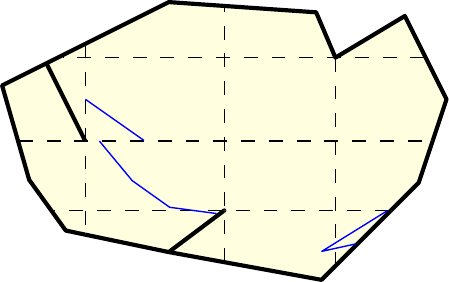}
    {\small \caption{A nonconvex cell $c$, its subdivision, and the (blue) curves $\Seg{c}{\tau}$.}
    \label{fig:linesegments}}
\end{wrapfigure}
The definition of an $\eps$-placement remains the same when we place a square amidst line segments,
and again we study the complexity of the $\eps$-placement space.
In a nonconvex cell $c$ of the arrangement, the curve $\Seg{c}{\tau}$ is no longer a piecewise linear convex
curve, but it can have several disconnected pieces; see Figure~5. 


\begin{theorem}
Given a set of line segments $\T$, the arrangement of $\eps$-placements of a unit square has complexity $\Oh(\frac{n^2}{\eps^2})$.
\end{theorem}

\begin{proof}
For analysis purposes, we imagine that any nonconvex cell $c$ of $\arr(\T)$ is partitioned into convex subcells
as follows.
For each endpoint of a line segment, shoot a ray in each orthogonal direction inside the cell
until it hits another line segment.
Use these rays to subdivide the cell into convex subcells, see Figure~5. 
If a cell $c$ has $k$ endpoints in its boundary, it is partitioned into $\Oh(k^2)$ subcells.
Within each subcell of $c$, the curve $\Seg{c}{\tau}$ has the same properties as in a
convex cell, and hence we can apply the same arguments as before.
The total number of endpoints is $\Oh(n)$, and hence the total number of subcells analyzed
in the arrangement $\arr(\T)$ is $\Oh(n^2)$. The bound follows.
\end{proof}

\subsection {Circle on lines}

We now place a unit circle, $\Ci_p$, on $\Li$, rather than a square.
We again define a set of translation vectors $T_\eps$ to points $\eps$-spaced on the boundary of $\Ci_p$, this time assuming $2\pi/\eps$ is an integer.
Consider cell $c$ in the arrangement $\arr$. The segments in $S(c,\tau)$ are no longer straight, but are generally elliptic.

\begin{figure}[t]
\begin{center}
\includegraphics[scale=0.7]{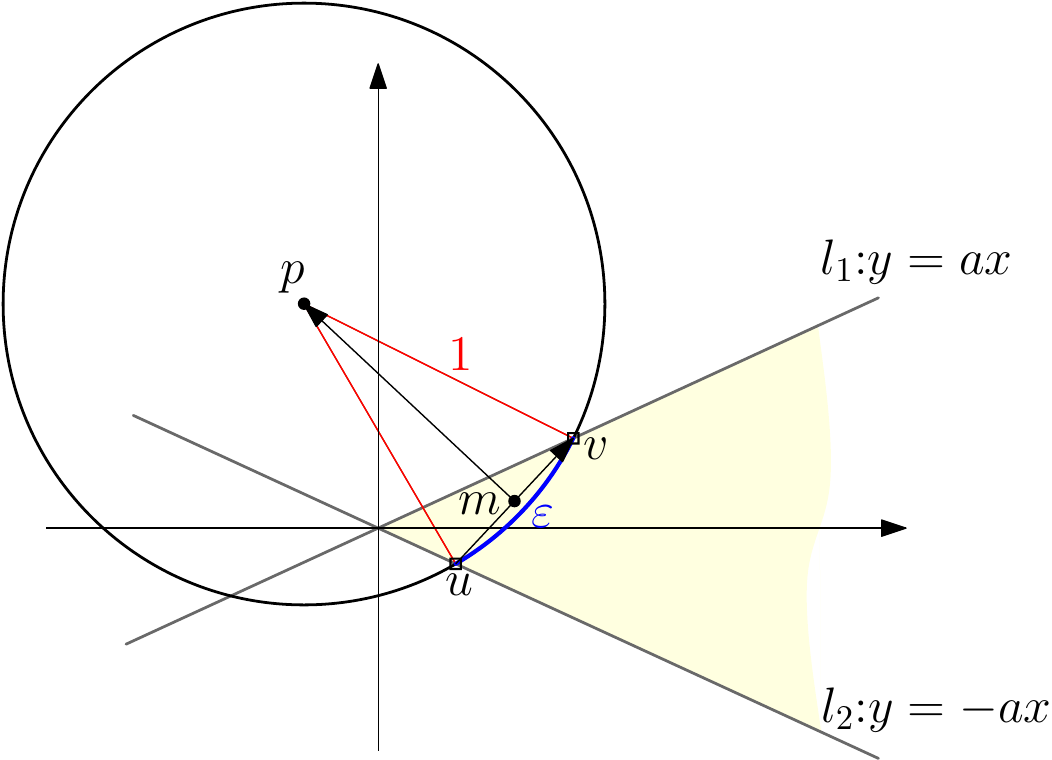}
\end{center}
\caption{Two lines $\ell_{1},\ell_{2}\in\Li$ and an $\eps$-placement of $p$ such that the corresponding $\eps$-segment is in the highlighted region.}
\label{fig:ellipse-deriv}
\end{figure}

\begin {lemma}
The $\eps$-placements of a circle such that the corresponding $\eps$-segments lie in $c$ form a collection of segments and elliptic arcs.
\end {lemma}
\begin {proof}
Consider two lines $\ell_{1}$ and $\ell_{2}\in\Li$.
Consider a coordinate system with the origin in the intersection point of $\ell_{1}$ and $\ell_{2}$ and oriented such that $\ell_{1}$ has equation $y=ax$, and $\ell_{2}$ has equation $y=-ax$.
Consider an $\eps$-placement of a unit circle such that the $\eps$-segment is in the I- and IV-quadrants below line $\ell_{1}$ and above $\ell_{2}$ as in Figure~\ref{fig:ellipse-deriv}.

Let $u$ be the intersection point of the circle and line $\ell_{2}$, $v$ be the intersection of the circle and $\ell_{1}$, and $m$ be a middle point between $u$ and $v$. From the following equations:
\[
\begin{aligned}
|uv|^{2}&=2-2\cos{\eps}\,,\\
|up|&=1\,,\\
|vp|&=1\,,\\
\vec{up} \times \vec{mp}&=|up|\cdot |mp|\,,
\end{aligned}
\]
we can derive that the segment of the $\eps$-placement curve that corresponds to all such $\eps$-segments that have one end point on line $\ell_{1}$ and another end on $\ell_{2}$ is an arc of an ellipse given by the following equation (refer to Figure~\ref{fig:ellipse}):
\[
\frac{a^2 x^2}{\left(\sin \left(\frac{\varepsilon }{2}\right)-a \cos \left(\frac{\varepsilon }{2}\right)\right)^2}+\frac{y^2}{\left(a \sin \left(\frac{\varepsilon }{2}\right)+\cos \left(\frac{\varepsilon }{2}\right)\right)^2}=1\,.
\]
\mbox{ }
\end {proof}
From this analysis we see that a sequence of adjacent arcs is not necessarily convex:
when $a>\tan{\frac{\eps}{2}}$ the point $p$ is tracing the left arc of an ellipse (the convex part), when $a<\tan{\frac{\eps}{2}}$ point $p$ is tracing the right arc of an ellipse (the concave part), and when $a=\tan{\frac{\eps}{2}}$ the ellipse degenerates to a straight segment.
Nonetheless, we can show that this cannot happen arbitrarily often.

\begin {lemma}  \label {lem:constantconvexpieces}
  For $\eps < 1$, $S(c,\tau)$ consists of a constant number of piece-wise elliptic convex curves.
\end {lemma}
\begin {proof}
  Consider the sequence of angles $a_1, a_2, \ldots, a_k$ between consecutive edges of $c$. Because $c$ is convex, we know that $\sum_{i=1}^k(\pi-a_i) = 2\pi$. For sufficiently small $\eps$ ($\eps < 1$ suffices), this implies that no more than two angles $a_i$ can be smaller than $\pi - \tan{\frac{\eps}{2}}$.

  By the previous lemma, these at most two angles correspond to concave elliptic arcs, which we report as separate curves. The remainder of the arcs and segments formed by boundary of $c$ is subdivided by these gaps into at most two piece-wise elliptic convex curves.

  As in the case for squares (refer to Lemma~\ref {lem:convex}), at most two disconnected curves, each of which can be decomposed into at most four convex subcurves, may exist for any given $\tau$.
  To see this, consider any line with direction $\tau$ that intersects $c$ in a segment $s$.
  As we slide $\tau(p)$ along $s$, the length of the boundary piece of $\Ci_p$ containing the point $\tau(p)$ in $c$ changes as a concave function. Hence, the piece is an $\eps$-segment at most twice.
  (Note that, in contrast to the square case, it is essential that $s$ has orientation $\tau$.)
\end {proof}

\begin{figure}[t]
\begin{center}
\includegraphics[scale=.45]{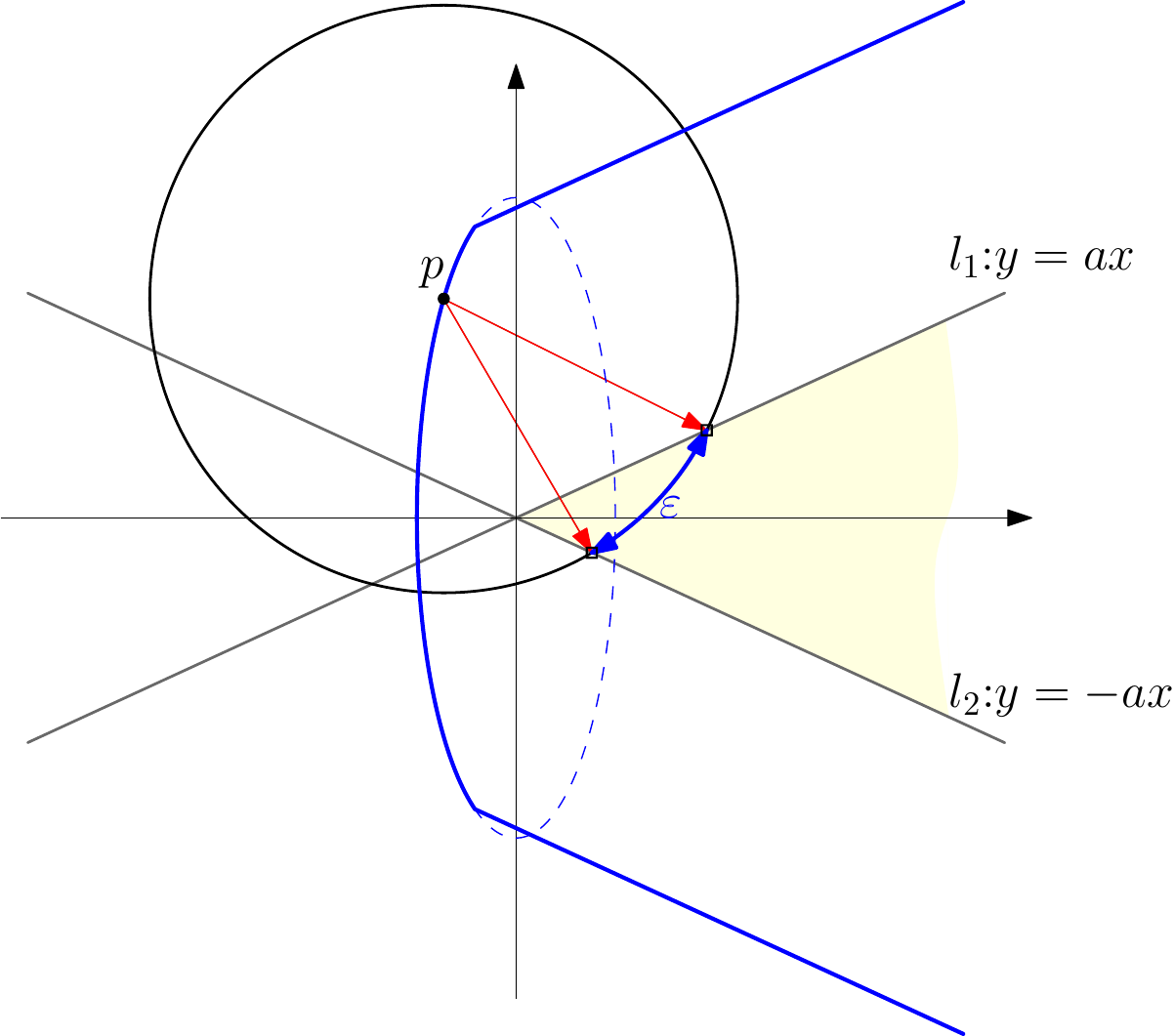}
\hfill
\includegraphics[scale=.45]{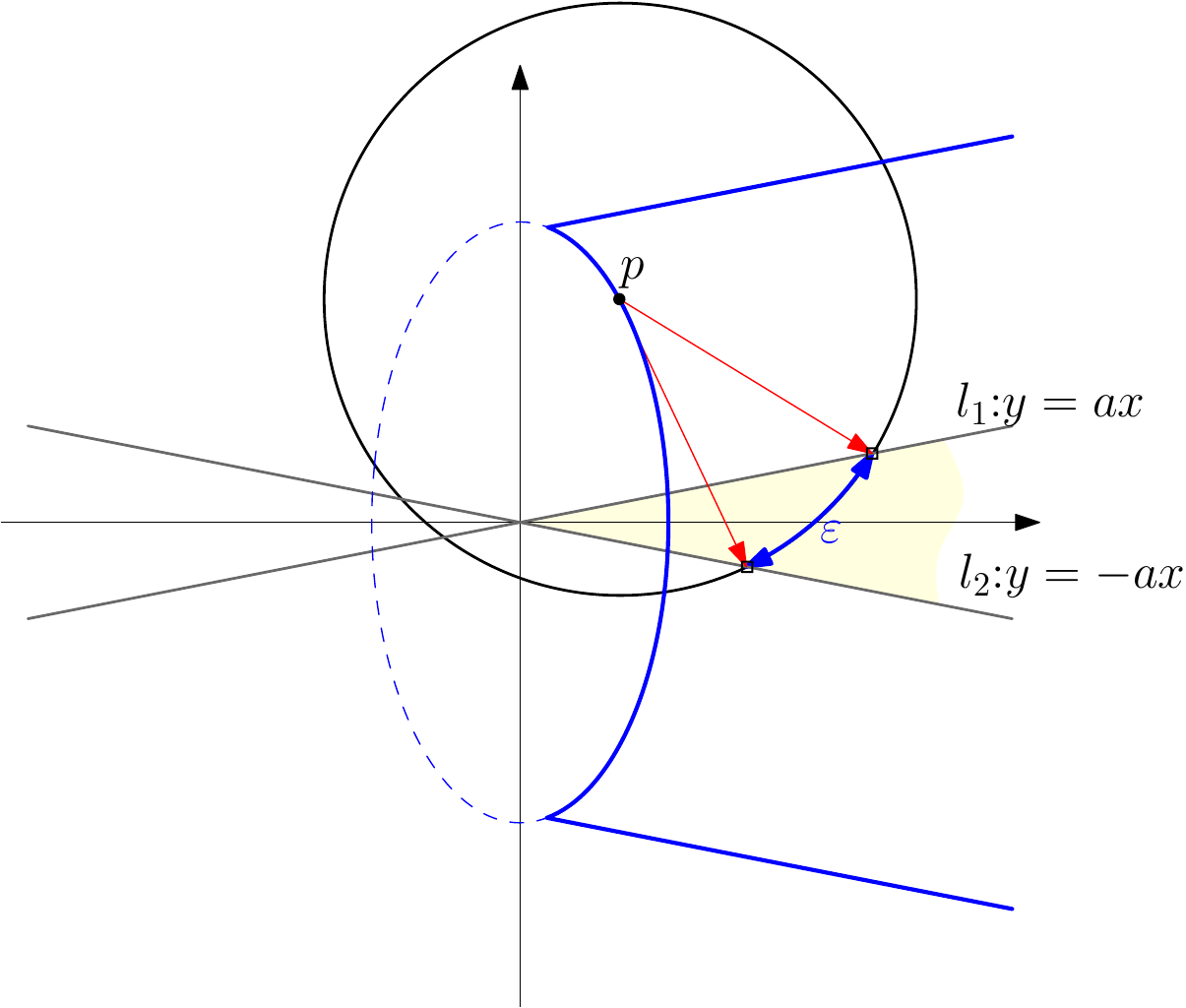}
\end{center}
\caption{Two lines $\ell_{1},\ell_{2}\in\Li$ and the curve (blue) that $p$ traces over all $\eps$-placements such that the corresponding $\eps$-segments are within the highlighted region.}
\label{fig:ellipse}
\end{figure}


To bound the complexity of the placement space, it is sufficient to bound the number of intersection points between $S(\tau)$ and $S(\sigma)$.

\begin{lemma}
For distinct $\tr, \sigma \in T_\eps$, the intersection of $\Segs{\tr}$ and $\Segs{\sigma}$ contains $\Oh(n^2)$ points.
\end{lemma}

\begin {proof}

Let $\arr'$ be the overlay of $\tr^{-1}(\arr)$ and $\sigma^{-1}(\arr)$: the arrangement of two copies of $\Li$ translated
by the inverses of $\tau$ and $\sigma$. Let $c$ be a cell in $\arr'$,
let $c_\tr$ and $c_\sigma$ denote the cells in $\arr$ which respectively
contain $\tr(p)$ and $\sigma(p)$. Observe that $c = \tr^{-1}(c_\tr) \cap \sigma^{-1}(c_\sigma)$.

Now, by Lemma~\ref {lem:constantconvexpieces}, $c_\tr$ and $c_\sigma$ both contain a constant number of (pieces of) convex curves. Each curve piece itself may consist of many elliptic arcs.
We observe that a convex curve, consisting of $k$ elliptic arcs, and a second convex curve, consisting of $h$ elliptic arcs, may cause at most $\Oh(k+h)$ intersections. We thus charge the intersections to the pieces of $S(\tau)$ or $S(\sigma)$, and note that each piece is charged at most constantly often.
\end {proof}

As in the case of squares, we have $\Oh(1/\eps^2)$ distinct translation vectors, so we can bound the total complexity by $\Oh(n^2/\eps^2)$.

\begin{theorem}
Given a set of $n$ lines $\Li$, the arrangement of $\eps$-placements of a unit circle
has complexity $\Oh(\frac{n^2}{\eps^2})$.
\end{theorem}

\section{Lower bounds}
\label {sec:lower}

\begin{figure}[t]
\begin{center}
\includegraphics[width=0.4\columnwidth]{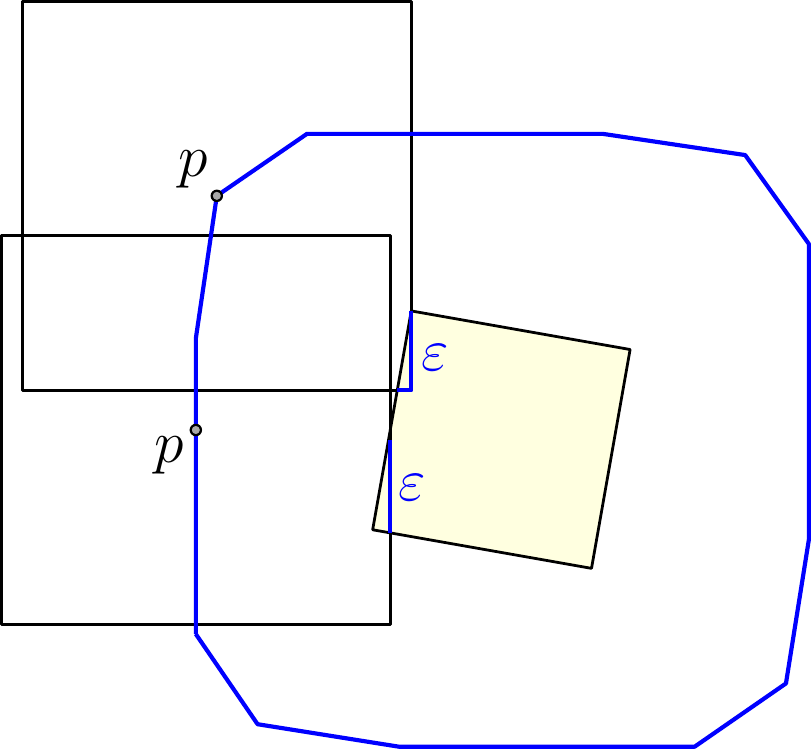}
\hfil
\includegraphics[width=0.4\columnwidth]{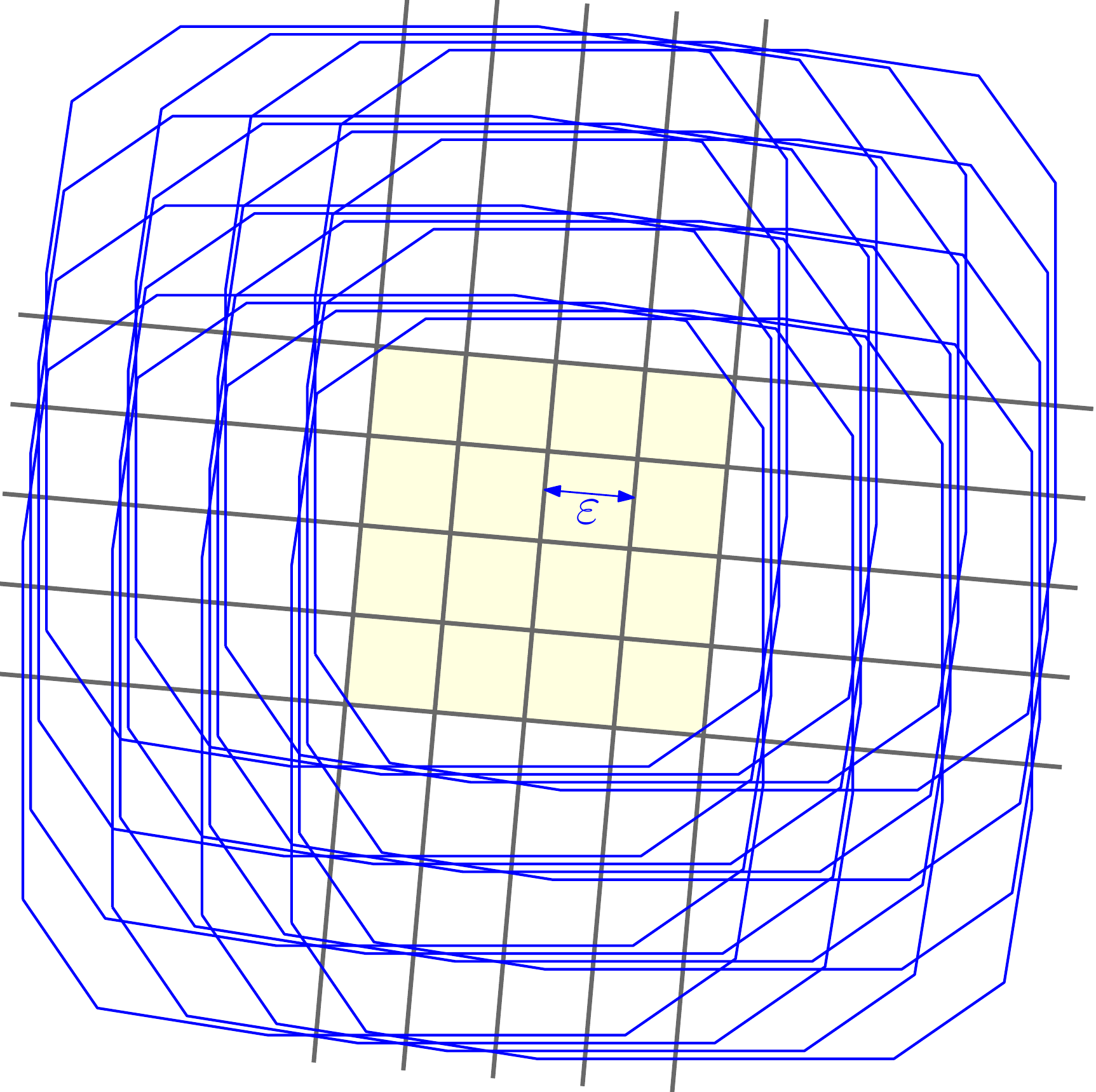}
\end{center}
\caption{
Lower bound construction for squares.
(left) Point $p$ traces an ``offset'' around square cell.
(right) A $\frac{n}{2}\times \frac{n}{2}$ grid formed by lines from $\Li$. Each cell has size $\Oh(\eps)$.}
\label{fig:square-lower-bound}
\end{figure}


By a worst case construction we show a tight bound on the complexity of the placement space of a square (the construction for placing a circle is analogous).
Consider a (slightly tilted) square cell of size $\Oh(\eps)$. The curve traced by $p$ of all $\eps$-placements of a unit side square such that an $\eps$-segment is inside the cell is shown in \rfig{square-lower-bound}.
It forms an ``offset'' curve around the cell of width $\approx 2$.
Place $\frac{n}{2}$ almost horizontal lines and $\frac{n}{2}$ almost vertical lines to form a grid with cells of size $\Oh(\eps)$ (consider all lines slightly tilted). 
\rfig{square-lower-bound} shows this construction.

%

A trivial upper bound\footnote{If $\eps$ is arbitrarily small,
the arrangement of the placement space is a collection of $n^2$ unit squares which can intersect pairwise.} on the complexity of the placement space is $\Oh(n^4)$, which by our
construction is worst case tight if $\eps = \Oh(\frac{1}{n})$.
If $\eps$ is bounded from below, the parametrized complexity of the arrangement is also worst-case tight.

\begin{cor}
Given a set of $n$ lines $\Li$ and a parameter $\eps = \Omega(\frac{1}{n})$, the arrangement of $\eps$-placements has a worst case complexity $\Omega(\frac{n^2}{\eps^2})$.
\end{cor}

\section{Computation}

To compute the critical placement of a square on $n$ lines $\Li$, we first compute $\arr(\Li)$ in $O(n^2)$ time.
Next, we traverse all cells $c$ in $\arr$ and calculate $\cup_{\tr \in T_\eps} \Seg{c}{\tr}$.
We do this by scanning the boundary of $c$, taking linear time in the complexity of $c$;
therefore, in total we spend $\Oh(n^2)$ time finding all $\Oh(n^2)$ critical segments.
Given this set of line segments we can compute their arrangement with standard techniques \cite{chazelle1992optimal,h-a-97,mulmuley},
giving an $\Oh(n^2 \log n + k)$ running time for output size $k$.
\begin{theorem}
An arrangement of critical placements of a unit axis-aligned square among a set of $n$ lines $\Li$ can be computed in $\Oh(n^2 \log n+k)$ time, where $k$ is the output size.
\end{theorem}

\noindent
{\bf Extensions.}
When placing a square over line segments $\T$, we first compute $\arr(\T)$.
With this arrangement we can find the subdivisions as in \rfig{linesegments},
and compute $\cup_{\tr \in T_\eps} \Seg{c}{\tr}$ for subdivisions $c$ by
scanning their convex boundary.
\begin{theorem}
An arrangement of critical placements of a unit axis-aligned square among a set of $n$ line segments $\T$ can be computed in $\Oh(n^2 \log n+k)$ time, where $k$ is the output size.
\end{theorem}

For placing a circle on lines, we compute $\cup_{\tr \in T_\eps} \Seg{c}{\tr}$
per cell $c$ of $\arr(\Li)$, spending $\Oh(n^2)$ time to find all arcs. Given a set of $n$ Jordan arcs, the complexity of building an arrangement is $\Oh((n+k)\log n)$~\cite{h-a-97}.\
Therefore, in our case we get:
\begin{theorem}
An arrangement of critical placements of a unit circle among a set of $n$ lines $\Li$ can be computed in $\Oh((n^2+k) \log n)$ time, where $k$ is the output size.
\end{theorem}

\section{Applications to junction detection}
\label{sec:junction}

While junctions are normally features of (road) networks, our objective is to extract
junctions purely based on trajectory data. This allows us to identify regions that
serve as junctions in open spaces like city squares, or to identify places where animals
pass by and choose one or the other direction.

We describe basic properties of a junction, which motivate
corresponding definitions.
Our approach is to treat any point in the plane as a potential junction; we define
when a point is \emph{junction-like} depending on the trajectories in its neighborhood.
We show that the arrangement from \rsec{complex} can be used to partition the plane
into regions of points that are similarly junction-like, that is to say they have
the same basic properties.
In the appendix we present---as a proof of concept---the output of a prototype 
implementation run on idealized data.
\medskip

\noindent
{\bf Properties of a junction.}
A junction is a region where trajectories enter and leave in a limited number of directions or routes,
called \emph{arms} of the junction. While the moving entities initiate a direction to
leave the junction in the core area, the arms are only ``visible'' somewhat further away
from the core; see \rfig{junctdefinition}.
This motivates the use of two concentric squares to decide whether a point is junction-like.
The way in which the trajectories enter and leave the two squares determines
whether the point is junction-like and how significant that junction is.
Since we can scale the plane with all trajectories, we can assume that the larger
square is unit-size. We let the smaller square have side length $\tfrac{1}{2}$.

\begin{figure}[tb]
\begin{center}
\includegraphics[width=.9\textwidth]{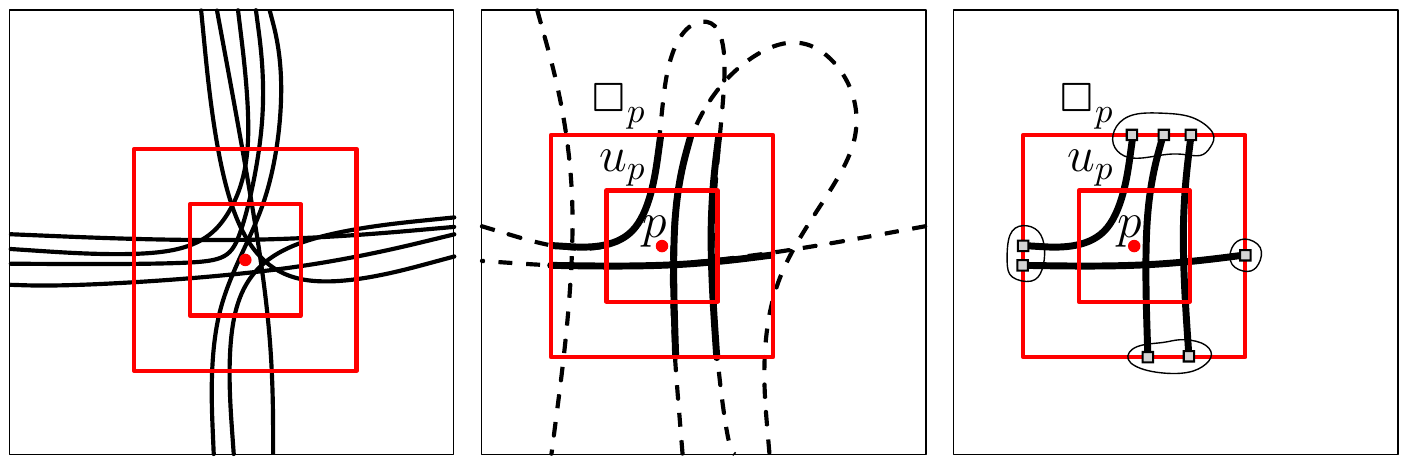}
\end{center}
\caption{Left: Junction-like point with its two concentric squares.
Middle: Solid subtrajectories are salient for the junction.
Right: $\eps$-clustering of the salient subtrajectory endpoints, revealing four arms. }
\label{fig:junctdefinition}
\end{figure}

For a point $p$ in the plane, let $u_p$ and $\Sq_p$ denote the boundaries of squares with
side length $\tfrac{1}{2}$ and $1$, respectively, both centered at $p$.

\begin{defi}
A (sub)trajectory $T$ is \emph{salient} for a point $p$ if it lies completely inside $\Sq_p$,
it intersects $u_p$, and its endpoints lie on $\Sq_p$.
\end{defi}

From a set $\T$ of trajectories we consider all subtrajectories that are salient;
see \rfig{junctdefinition}.
Such subtrajectories should be maximal: their endpoints must be actual crossings with
the boundary of $\Sq_p$ and not just contacts.
A single trajectory in $\T$ can have multiple salient subtrajectories.

\begin{defi}
Given a set of points $Q$ on the boundary of a square $\Sq_p$ and a constant $\eps>0$, an
\emph{$\eps$-clustering} is a partitioning of $Q$ such that for any two points $q,q'\in Q$ belonging to the same cluster,
there exists a sequence of points $q=q_1,\ldots,q_j=q'$, all in $Q$, such that $q_i$ and
$q_{i+1}$ for $1\leq i < j$ have distance at most $\eps$ along $\Sq_p$.
\end{defi}

\begin{defi}
Let $p$ be a point in the plane, let $\T$ be a set of trajectories, let $\eps>0$, and
let $Q_p(\T)$ be the set of endpoints of all salient subtrajectories from $\T$.
A point $p$ is \emph{junction-like} if an $\eps$-clustering of $Q_p(\T)$ has
at least three clusters.
\end{defi}

When we move the point $p$ over the plane, these clusters can grow, shrink, merge
and split. A cluster may for instance split when two consecutive intersection points
from $Q_p(\T)$ become more separated than $\eps$.
A cluster may shrink because a subtrajectory is no longer salient,
which then may also cause a split of a cluster.

It should be clear that we can compute a subdivision of the plane into
maximal cells where the $\eps$-clustering is the same. It should also be clear that
the theory of \rsec{complex} discusses a simplified version of this problem where
salience is not considered.
\medskip

\noindent
{\bf Implementation and results.}
There are many ways to obtain junctions from the junction-like points and their $\eps$-clusterings.
There are also many ways to attach a significance to a junction, based on the number of clusters
and their sizes.
Moreover, we may want to distinguish between \emph{real junctions} and \emph{crossings},
where a crossing is a junction-like region with four arms where all trajectories go
straight, and a real junction has several different splits and merges of trajectories
over the clusters. For both types we can define their significance in various ways;
see \cite{vanduijn} for further discussions.

\rfig{vis} shows some results of an implementation that samples points from a regular square grid
and evaluates the significance of a junction according to such a measure. Junction-like
points are indicated by a colored grid cell, which contains a measured value. 
What can be observed is that an area around each junction contains several junction-like
points with the same measure, and that the measure drops quickly towards the edge of this region.
It is easy to
convert the output so that for each group of junction-like cells, only one junction is reported.
\rfig{vis2} shows that this type of definition can indeed identify junctions and assign a 
significance in a reasonable manner.

\begin{figure*}[htb]
\begin{center}
\includegraphics[width=.9\textwidth]{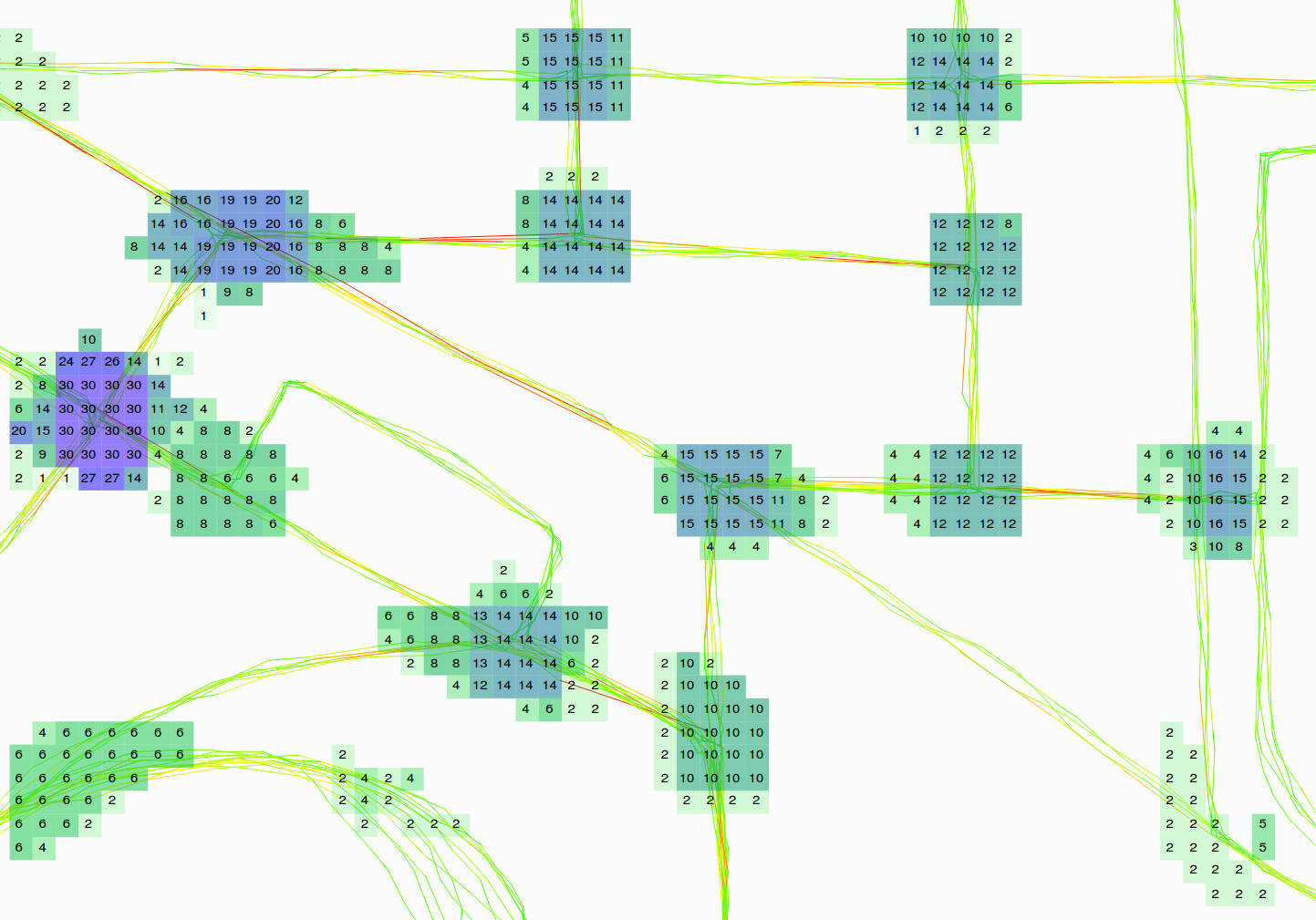}
\end{center}
\caption{Junction-like points indicated by colored grid cells, darker cells contain a higher measure.}
\label{fig:vis}
\end{figure*}

\begin{figure*}[htb]
\begin{center}
\includegraphics[width=.9\textwidth]{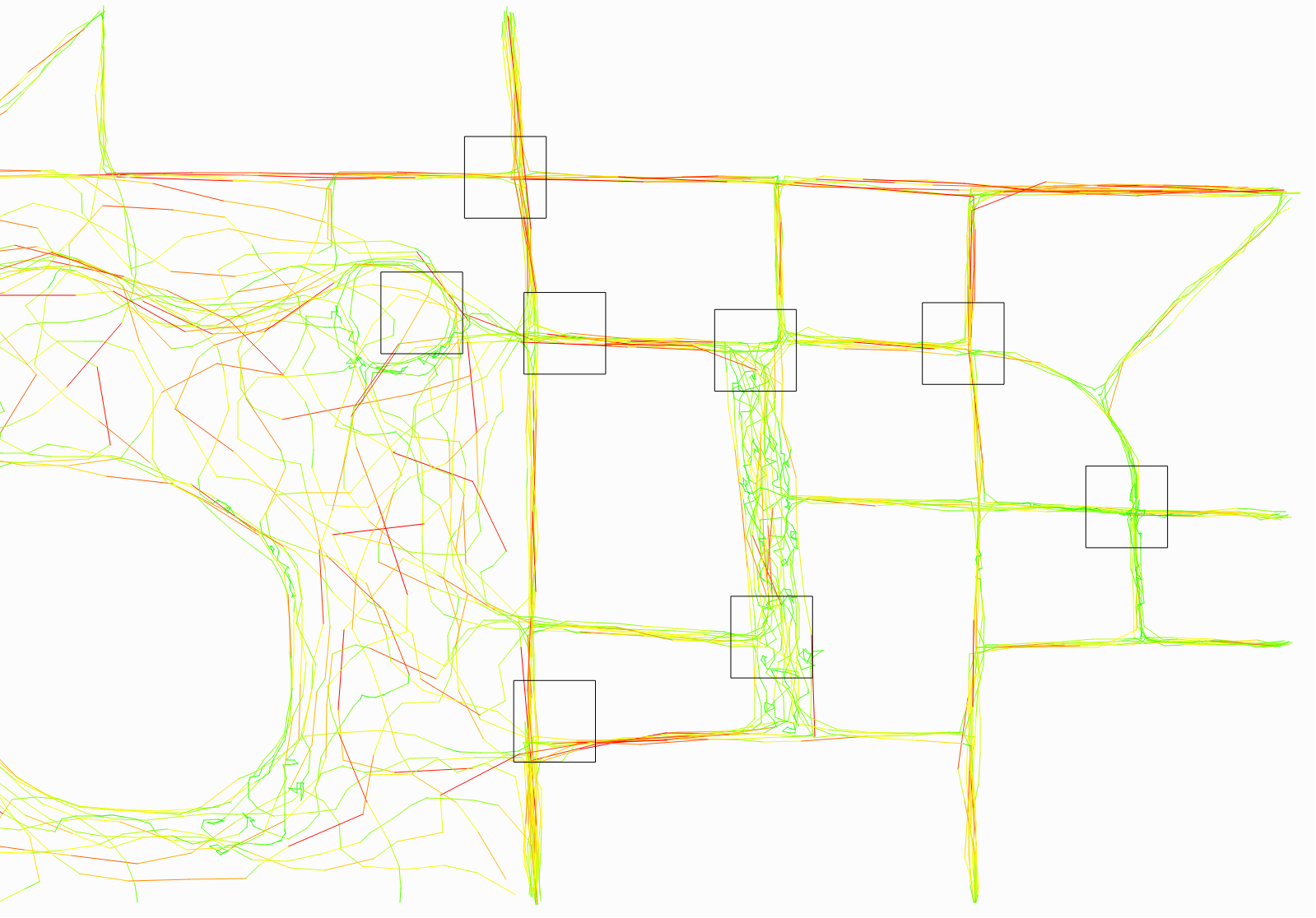}
\end{center}
\caption{Locations of $u_p$ indicated for the eight most significant junction-like points.}
\label{fig:vis2}
\end{figure*}


\section{Conclusions and further research}
\label {sec:conclusions}

We analyzed the complexity of the placement space of a unit square or circle in an arrangement of lines or line segments,
a problem that lies at the core of junction detection in trajectory analysis.
Our results include upper bounds that improve on the naive $O(n^4)$ bound by showing that two of the linear factors in fact only need to depend on $1/\eps$.
Increasing the number of trajectories is not a reason to increase the granularity of the clustering, so in practice we expect $1/\eps$ to be much smaller than $n$, if not constant.
The resulting $\Theta(n^2/\eps^2)$ bound is tight in the worst case.
The combinatorial problem is interesting from the theoretical perspective because it combines arrangements with distances in the arrangement.

A prototype implementation applies the result to junction detection.
While the implementation demonstrates the value of the approach and indicates that the time complexity is not prohibitive, several simplifications of the problem have been made and no extensive experiments have been performed.
In a follow-up study, it would be interesting to augment the implementation with circles or different shapes, and to run the algorithm on real-world trajectory data from various sources.

From a theoretical perspective, it would be interesting to explore further extensions of the approach (i.e. curved trajectories, trajectories embedded in higher-dimensional spaces),
or apply it to other distance based arrangements.

\newpage
\section*{Acknowledgments}

M.L. is supported by the Netherlands Organisation for Scientific
Research (NWO) under grant 639.021.123.
MADALGO Center for Massive Data Algorithmics is supported in part by the Danish National Research Foundation grant DNRF 84.

\newpage
\small
\bibliographystyle{plain}
\bibliography{junct}

\begin{thebibliography}{10}

\bibitem{BenkertDGW10}
M.~Benkert, B.~Djordjevic, J.~Gudmundsson, and T.~Wolle.
\newblock Finding popular places.
\newblock {\em Int. J. Comput. Geometry Appl.}, 20(1):19--42, 2010.

\bibitem{Buchin}
K.~Buchin, M.~Buchin, M.~van Kreveld, M.~L{\"{o}}ffler, J.~Luo, and R.~I.
  Silveira.
\newblock Processing aggregated data: the location of clusters in health data.
\newblock {\em GeoInformatica}, 16(3):497--521, 2012.

\bibitem{chazelle1992optimal}
B.~Chazelle and H.~Edelsbrunner.
\newblock An optimal algorithm for intersecting line segments in the plane.
\newblock {\em Journal of the ACM (JACM)}, 39(1):1--54, 1992.

\bibitem{glw-mpstd-08}
J.~Gudmundsson, P.~Laube, and T.~Wolle.
\newblock Movement patterns in spatio-temporal data.
\newblock In S.~Shekhar and H.~Xiong, editors, {\em Encyclopedia of GIS}, pages
  726--732. Springer, 2008.

\bibitem{gk}
J.~Gudmundsson and M.~van Kreveld.
\newblock Computing longest duration flocks in trajectory data.
\newblock In {\em 14th {ACM} International Symposium on Geographic Information
  Systems}, pages 35--42, 2006.

\bibitem{gks}
J.~Gudmundsson, M.~van Kreveld, and F.~Staals.
\newblock Algorithms for hotspot computation on trajectory data.
\newblock In {\em 21st {SIGSPATIAL} International Conference on Advances in
  Geographic Information Systems}, pages 134--143, 2013.

\bibitem{h-a-97}
D.~Halperin.
\newblock Arrangements.
\newblock In J.~E. Goodman and J.~O'Rourke, editors, {\em Handbook of Discrete
  and Computational Geometry}, chapter~24, pages 529--562. CRC Press LLC, Boca
  Raton, FL, 2004.

\bibitem{hm}
S.~Har{-}Peled and S.~Mazumdar.
\newblock Fast algorithms for computing the smallest k-enclosing circle.
\newblock {\em Algorithmica}, 41(3):147--157, 2005.

\bibitem{jyj-tpm-11}
H.~Jeung, M.~Yiu, and C.~Jensen.
\newblock Trajectory pattern mining.
\newblock In Y.~Zheng and X.~Zhou, editors, {\em Computing with Spatial
  Trajectories}, pages 143--177. Springer, 2011.

\bibitem{MorenoTRB10}
B.~Moreno, V.C. Times, C.~Renso, and V.~Bogorny.
\newblock Looking inside the stops of trajectories of moving objects.
\newblock In {\em Proc. XI Brazilian Symposium on Geoinformatics}, pages 9--20.
  MCT/INPE, 2010.

\bibitem{mount}
D.~M. Mount, R.~Silverman, and A.~Y. Wu.
\newblock On the area of overlap of translated polygons.
\newblock {\em Computer Vision and Image Understanding}, 64(1):53--61, 1996.

\bibitem{mulmuley}
K.~Mulmuley.
\newblock {\em Computational geometry - an introduction through randomized
  algorithms}.
\newblock Prentice Hall, 1994.

\bibitem{openshaw}
S.~Openshaw, M.~Charlton, C.~Wymer, and A.~Craft.
\newblock A mark 1 geographical analysis machine for the automated analysis of
  point data sets.
\newblock {\em International Journal of Geographical Information System},
  1(4):335--358, 1987.

\bibitem{PalmaBKA08}
A.T. Palma, V.~Bogorny, B.~Kuijpers, and L.~Ot{\'a}vio Alvares.
\newblock A clustering-based approach for discovering interesting places in
  trajectories.
\newblock In {\em Proc. 2008 ACM Symposium on Applied Computing}, pages
  863--868, 2008.

\bibitem{TiwariK13}
S.~Tiwari and S.~Kaushik.
\newblock Mining popular places in a geo-spatial region based on {GPS} data
  using semantic information.
\newblock In {\em Proc. 8th Workshop on Databases in Networked Information
  Systems}, volume 7813 of {\em LNCS}, pages 262--276, 2013.

\bibitem{vanduijn}
I.~van Duijn.
\newblock Pattern extraction in trajectories and its use in enriching
  visualisations.
\newblock Master's thesis, Department of Information and Computing Sciences,
  Utrecht University, 2014.
\newblock http://dspace.library.uu.nl/handle/1874/294075.

\bibitem{hk}
S.~van Hagen and M.~van Kreveld.
\newblock Placing text boxes on graphs.
\newblock In {\em Graph Drawing, 16th International Symposium}, volume 5417 of
  {\em Lecture Notes in Computer Science}, pages 284--295. Springer, 2008.

\bibitem{ksw}
M.~van Kreveld, {\'E}.~Schramm, and A.~Wolff.
\newblock Algorithms for the placement of diagrams on maps.
\newblock In {\em Proceedings of the 12th annual ACM international workshop on
  Geographic information systems}, pages 222--231. ACM, 2004.

\end{thebibliography}

\end{document}